\documentclass[letterpaper,twocolumn,american,showpacs,pr,aps,superscriptaddress,eRSecnum,nofootinbib]{revtex4}

\usepackage{amsfonts}
\usepackage{amsmath}
\usepackage{amssymb}
\usepackage{graphicx}
\usepackage{color}
\usepackage{hyperref}
\usepackage{comment}
\usepackage{subfig}
\newtheorem{theorem}{Theorem}

\newenvironment{proof}[1][Proof]{\noindent\textbf{#1.} }{\ \rule{0.5em}{0.5em}}

\definecolor{nblue}{rgb}{0.2,0.2,0.7}
\definecolor{ngreen}{rgb}{0.2,0.6,0.2}
\definecolor{nred}{rgb}{0.7,0.2,0.2}
\definecolor{nblack}{rgb}{0,0,0}

\newcommand{\bes} {\begin{subequations}}
\newcommand{\ees} {\end{subequations}}
\newcommand{\bea} {\begin{eqnarray}}
\newcommand{\eea} {\end{eqnarray}}
\newcommand{\beq}{\begin{equation}}
\newcommand{\eeq}{\end{equation}}

\def\>{\rangle}
\def\<{\langle}
\def\Tr{\textrm{Tr}}

\newcommand{\ignore}[1]{}

\definecolor{nblue}{rgb}{0.2,0.2,0.7}
\definecolor{ngreen}{rgb}{0.2,0.6,0.2}
\definecolor{nred}{rgb}{0.7,0.2,0.2}
\definecolor{nblack}{rgb}{0,0,0}

\begin{document}

\title{From clocks to cloners: Catalytic transformations under covariant operations and recoverability}


\author{Iman Marvian}
\affiliation{Research Laboratory of Electronics, Massachusetts Institute of Technology, Cambridge, MA 02139}
\author{Seth Lloyd}
\affiliation{Research Laboratory of Electronics, Massachusetts Institute of Technology, Cambridge, MA 02139}
\affiliation{Department of Mechanical Engineering, Massachusetts Institute of Technology, Cambridge, MA 02139}

\begin{abstract}

There are various  physical scenarios in which one can only implement operations with a certain symmetry. Under such restriction, a system in a symmetry-breaking state can be used as a  catalyst, e.g. to prepare another system in a desired symmetry-breaking state.   This sort of (approximate) catalytic state transformations are relevant in the context of  (i) state preparation using a bounded-size quantum clock or reference frame, where the clock or reference frame acts as a catalyst, 
(ii) quantum thermodynamics, where again a clock can be used as a catalyst to prepare states which contain coherence with respect to the system Hamiltonian, and (iii)  cloning of unknown quantum states, where the given copies of state  can be interpreted as a catalyst for preparing the new copies. Using a recent result of Fawzi and Renner on approximate recoverability, we show that the achievable  accuracy in this kind of catalytic transformations  can be determined by a single function, namely the relative entropy of asymmetry, which 
 is equal to the  difference between the entropy of state and its symmetrized version: if the desired state transition  does not require a large increase of this quantity, then it can be implemented with high fidelity using only symmetric operations.  Our lower bound on the achievable  fidelity is tight in the case of cloners, and can be achieved using the Petz recovery map, which interestingly turns out to be the optimal cloning map found by Werner. 

\end{abstract}
\date{\today}

\maketitle

\section{Introduction}
A wide  range of seemingly different physical and information theoretic problems can be understood and formalized as special cases of the following abstract problem: we want to transform a known input state  to a desired output state, or a state close to it, under the restriction that we can only implement operations which respect a certain symmetry.  As a result of this restriction on the allowed operations, many state transformations are not possible. The simplest  example is the impossibility of transforming a symmetric state to a symmetry-breaking one, which is reminiscent of the Curie's principle.  The relevant property of states that determines whether such transformations are possible or not, can be interpreted as the symmetry-breaking, or the \emph{asymmetry} of state relative to the symmetry under consideration. The resource theory of asymmetry is a framework for classifying and quantifying this property, and answering this sort of questions about state transitions under symmetric operations  \cite{gour2008resource, Asymy2, MS_Short, Noether}.

In this paper we consider a special case of this general problem, namely \emph{catalytic} transformations under symmetric operations. More specifically, we consider the case where a system in a symmetry-breaking state is used as a \emph{catalyst} to transform another system in an initial symmetric state to a desired  asymmetric state, such that the catalyst remains (approximately) unchanged. We argue that this scenario naturally arises  in the context of (i) state preparation using a bounded size quantum clock or  quantum reference frame \cite{QRF_BRS_07},  (ii) quantum thermodynamics  \cite{brandao2013resource}, and (iii)  cloning of unknown quantum states \cite{werner1998optimal}. Then, using a recent result of Fawzi and Renner \cite{fawzi2015quantum} on approximate recoverability, we show that the achievable  accuracy in this kind of catalytic transformations  can be determined by a single measure of asymmetry, namely the \emph{relative entropy of asymmetry}, which turns out to be equal to the  the entropy of the symmetrized version of state minus the entropy of the original state. If this quantity does not increase considerably in the desired state transformation  then the transformation can be implemented  with high fidelity using only symmetric operations. As we explain later, this is somewhat surprising, because for  general state transformations under symmetric operations, just by looking to a single measure of asymmetry one cannot guarantee  the possibility of a state transition, or an approximate version of it. We discuss some applications of this result in various contexts. In particular, we show that in the case of cloning  of quantum states our lower bound on the achievable fidelity can be achieved using the Petz recovery map \cite{petz1988sufficiency, petz1986sufficient}, which interestingly  turns out to be the optimal cloning map  originally  found  by Werner \cite{werner1998optimal}. 

The article is organized as follows. In Sec. \ref{Sec:Catal} we first provide a formal presentation of  the problem of catalytic transformations under symmetric operations, and then we show examples of different physical scenarios which all can be formalized as special examples of this problem. Then, in Sec. \ref{Sec:meas} we review the concept of measures of asymmetry, and in particular, we discuss about the \emph{relative entropy of asymmetry,} which is the measure of asymmetry we are concerned with in this paper. In Sec. \ref{Sec:upper} and Sec. \ref{Sec:lower} we present   upper and lower bounds on achievable accuracy of catalytic transformations in terms of the increase of the relative entropy of asymmetry in the desired transformation. In Sec. \ref{Sec:app} we present applications of these bounds in the context of state preparation using bounded size clocks, and cloning of unknown quantum states.  Finally, in Sec. \ref{Sec:revers} we present a generalization of our lower bound in Sec \ref{Sec:lower}, which shows   if during a symmetric process the relative entropy of asymmetry does not drop considerably, then the process can be approximately inverted using a symmetric operation.

\section{Catalytic Transformations }\label{Sec:Catal}

 In this paper we study (inexact) catalytic state transformations  in the following form 
\beq\label{main2}
\tau^R\otimes \rho_{\text{sym}}^S\xrightarrow{\text{Covariant}} \tau^R\otimes \sigma^S \ .
\eeq
Here $\tau^R$ and $\rho_{\text{sym}}^S$ are initial states of systems $R$ and $S$, which we sometimes refer to them as the \emph{reference frame} and \emph{system} respectively, and $\tau^R\otimes\sigma^S$ is the desired joint final state of these systems.  We assume $\rho_{\text{sym}}^S$ is invariant under the action of a certain symmetry. Then, by applying a time evolution which also has this symmetry, or is \emph{covariant} with respect to this symmetry,   we try to map the initial state $\tau^R\otimes \rho_{\text{sym}}^S$ to a state close to the desired  state $\tau^R\otimes \sigma^S$.

The most general physical time evolution that can be implemented on a quantum system (with probability one) can be described by a trace-preserving completely positive linear map, also known as \emph{quantum channel}. This includes unitary transformations, measurements, and adding systems or discarding subsystems. We consider quantum channels which are covariant  with respect to an arbitrary symmetry. Let $G$ be the group corresponding to the symmetry under consideration, $U^X_g$ denote the unitary representation of the group element $g\in G$ on system $X$, and $\mathcal{U}^X_g(\cdot)=U^X_g(\cdot){U_g^X}^\dag$ be the corresponding super-operator. Then, under the action of the group element $g\in G$ state $\rho^X$ of system $X$ transforms to  $\mathcal{U}^X_g(\rho^X)$. 
The representation of symmetry on the joint system $X_1$ and $X_2$  is given by $\mathcal{U}^{X_1X_2}_g=\mathcal{U}^{X_1}_g\otimes \mathcal{U}^{X_2}_g$.  As we see in the following examples, the form of representation  is often dictated by the physical interpretation of the symmetry. 

Then,  state $\rho^X$ is symmetric with respect to the symmetry under consideration, or \emph{G-invariant}, if for all $g \in G$, $\mathcal{U}^X_g(\rho^X)=\rho^X$. Similarly, a quantum channel   $\mathcal{E}^{X\rightarrow Y}$ from system $X$ to system $Y$ is called \emph{G-covariant}, or symmetric, if it commutes with the action of group, i.e.
\beq\label{cov}
\mathcal{E}^{X\rightarrow Y}\circ \mathcal{U}^X_g=\mathcal{U}^{Y}_g\circ \mathcal{E}^{X\rightarrow Y}\ \  , \ \ \ \ \  \forall g\in G\ .
\eeq
Note that to specify the set of G-covariant channels from system $X$ to system $Y$ we need to know the representation of symmetry $G$ on both spaces.   
We often suppress the superscript corresponding to the label of the system, when there is no chance of confusion. Also, since group $G$ is often clear from the context, we sometimes use the term \emph{covariant} or symmetric instead of $G$-covariant.  


We are interested to determine whether using only covariant operations, one can implement  the desired state transition in Eq.(\ref{main2}), or an approximate version of it.  But since, by assumption  state $\rho_\text{sym}^S$  in the input is invariant under the action of the group, i.e. $\mathcal{U}^S_g(\rho_\text{sym}^S)=\rho_\text{sym}^S$ for all $g\in G$, it turns out that the answer to this question is independent of this state. This is because for any state $ \tau^R$ and symmetric state  $\rho_\text{sym}^S$, there exists covariant channels which transform state $ \tau^R$ to $ \tau^R\otimes \rho_\text{sym}^S$, and vice versa. Therefore,  the  transformation in Eq.(\ref{main2}) can be implemented with covariant channels with certain error if and only if the state transition 
\beq\label{main}
\tau^R\xrightarrow{\text{Covariant}} \tau^R\otimes \sigma^S \ 
\eeq
can be implemented with the same error. Hence, in the following we sometimes consider this more concise form of the problem.

Unless state $\sigma^S$ in the output is also a symmetric state,  in general, using only covariant operations the state transition  $\tau^R\rightarrow \tau^R\otimes\sigma^S$  cannot be implemented exactly. Roughly speaking, this is because covariant channels cannot generate asymmetry (relative to the symmetry under consideration), and therefore their output state cannot contain more asymmetry than their input.  Later we see a more precise and quantitative version of this statement in terms of \emph{measures of asymmetry}.

To explain  the physical relevance of the catalytic transformations described in Eqs.(\ref{main2}, \ref{main}), in the following we consider some illustrative examples. 

\subsection{Clocks}
Catalytic transformations described above provide a natural framework for understating how clocks are used in a state preparation  process.  Suppose we want to prepare the system $S$ with Hamiltonian $H^S$ in state $\sigma^S$ which contains coherence in the energy eigenbasis. Any such state is time-dependent, and therefore the above statement about the state of system $S$ is ambiguous, unless we say $\sigma^S$ is the state of system at certain time $t$ with respect to a clock $R$. 
The clock could be, for instance, a Harmonic Oscillator oscillating at a certain frequency, or a free  particle moving on a line. Then, to prepare the system $S$ in state $\sigma^S$ we need to directly or indirectly interact it with the clock $R$. But, there is a restriction on the set of operations that we can implement on the system and clock: If the reference for time is defined just by the clock $R$ and the rest of systems do not have any information about how the time reference is defined,  then the only operations we can perform on the system $S$ and clock $R$ are those which do not explicitly or implicitly depend on the time reference. More precisely, these are operations which are covariant with respect to the group of time translations, i.e.  they  satisfy
\beq\label{TI}
\forall t\in \mathbb{R}:\  \mathcal{E}(e^{-i H_{\text{tot}} t }(\cdot)e^{i H_{\text{tot}} t } )=e^{-i H_{\text{tot}} t }\mathcal{E}(\cdot)e^{i H_{\text{tot}} t }\ ,
\eeq
where $H_{\text{tot}}=H^R+H^S$ is the total Hamiltonian of system and clock, and we have suppressed ${RS\rightarrow RS}$ superscript (We assume $\hbar=1$ throughout the paper).


Therefore, in this context  Eq.(\ref{main2}) has the following interpretation: to prepare the system $S$ in state $\sigma^S$ we couple it to a clock $R$, such that at the moment where the clock is in state $\tau^R$ the system is in state $\sigma^S$. Note that ideally the clock should play the role of a catalyst, i.e. it should  remain unaffected in the process.


The restriction to the set of channels which are invariant under time translations  also arises in the context of quantum thermodynamics,  where the only \emph{free} unitaries and states are,  respectively, energy-conserving unitaries and thermal states \cite{brandao2013resource, lostaglio2015quantum, lostaglio2015description}. Since these unitaries and states are all invariant under time-translations, it follows that the only free quantum operations in this resource theory  are those which are invariant under time-translations, i.e. satisfy Eq.(\ref{TI}). 
Again, in this context a clock, i.e. a system in a state that breaks the time-translation symmetry can be used as an approximate catalyst to perform certain transformations \cite{Aberg2014}.\footnote{Note that in the resource theory of thermodynamics, as defined in \cite{brandao2013resource}, systems could be in arbitrary non-equilibrium (time-dependent) states. Then, some state transitions are forbidden in this resource theory, simply because, e.g. the initial state is time-independent and the desired final state is time-dependent, and to prepare a time-dependent state one needs a clock, which is not allowed in this resource theory.  Indeed, this resource theory can be thought as the combination of two different resource theories: A resource theory for time-translation asymmetry (or \emph{unspeakable} coherence \cite{marvian2016quantify}), and a resource theory for energy and purity, where in the second resource theory all states are incoherent in the energy eigenbasis.  The latter resource theory seems to be more related to the standard thermodynamics, where all states are assumed to be close to \emph{equilibrium}.  }


\subsection{Reference Frames}

Similar to clocks, which are reference frames for time, other reference frames for other physical degrees of freedom can  be interpreted as catalysts.   For instance, when we use a quantum gyroscope to prepare another system in a direction defined by the gyroscope, we can only implement transformations which can be achieved using  isotropic operations, i.e. operations which are covariant with respect to the group of rotations $G=\text{SO}(3)$. Therefore, the task of  preparing system $S$ in a direction defined by the gyroscope  $R$ can be formally phrased as a catalytic transformation in the form of Eq.(\ref{main}). Another example, which is relevant in the context of quantum optics, is preparing states using  a phase reference for the phase conjugate to the photon number operator in a particular optical mode. Here, the relevant symmetry is the group of phase shifts, which is isomorphic to U(1) (See \cite{QRF_BRS_07} for an overview of quantum reference frames).    

\subsection{Cloners}

Another motivation for considering the catalytic  transformations in Eq.(\ref{main}) comes from the study of quantum cloning machines, or quantum \emph{cloners}: Suppose we are given $n\ge 1$ copies of an unknown state $\psi\in\mathbb{C}^d$, and we are interested to  generate $n+k$ copies of this state. The no-cloning theorem implies that for any $k>0$ the transformation $\psi^{\otimes n}\rightarrow \psi^{\otimes (n+k)}$  can  be implemented  only approximately  \cite{wootters1982single}. 

Therefore, to compare the performance of different cloners, we need to consider a figure of merit, such as fidelity of cloning. For most applications,  we expect that a good cloner should act equally well on all pure states, and hence the figure of merit should care equally about all pure states. This natural requirement implies that the optimal cloner $\mathcal{E}^{n\rightarrow n+k}$ from $n$ copies to $n+k$ copies can always be chosen to  satisfy the covariance condition
\beq\label{cloner}
\mathcal{E}^{n\rightarrow n+k}(U^{\otimes n}(\cdot) {U^\dag}^{\otimes n})=U^{\otimes {(n+k)}}\mathcal{E}^{n\rightarrow n+k}(\cdot) {U^\dag}^{\otimes (n+k)}\ ,
\eeq
for all $U\in \text{U}(d)$ \cite{werner1998optimal}. Note that  if a cloner satisfies this symmetry  then its performance on all pure states is uniquely specified by its performance on one pure state, because all pure states can be transformed to each other by unitary transformations. Therefore, assuming this symmetry holds, to evaluate the performance of the cloner on an unknown pure state  we can just focus on its action on one particular \emph{known} pure state $\psi$.  Then, in the desired transformation $\psi^{\otimes n}\rightarrow \psi^{\otimes (n+k)}$, we can interpret the given copies of state $\psi$ as catalyst. In other words, choosing $\tau^R=\psi^{\otimes n}$ and $\sigma^S=\psi^{\otimes k}$ in Eq.(\ref{main}), this equation describes the action of an ideal cloner. 

\section{Relative entropy of asymmetry}\label{Sec:meas}
A measure of asymmetry is a function that quantifies the amount of symmetry-breaking of states relative to a given symmetry \cite{Noether}. Formally, a function $f$ from states to real numbers is a measure of asymmetry with respect to a symmetry, if (i) it vanishes for all symmetric states, and (ii) it is non-increasing under G-covariant operations, i.e. $f(\mathcal{E}(\rho))\le f(\rho)$ for any state $\rho$ and G-covariant channel $\mathcal{E}$. Note that to define measures of asymmetry on a given Hilbert space, we need to know the representation of the symmetry on that space. 


A well-studied  measures of asymmetry is the \emph{relative entropy of asymmetry}, defined by $\Gamma(\rho)\equiv \inf_{\omega\in \text{sym}(G)}S(\rho\|\omega)$, where $S(\rho_1\|\rho_2)=\Tr(\rho_1(\log\rho_1-\log\rho_2))$ is the relative entropy, and  the minimization is over all states which are invariant under group $G$, i.e. $\mathcal{U}_g(\omega)=\omega$, for all $g\in G$ (Throughout the paper  we assume the base of logarithm is 2). Roughly speaking, this function quantifies the distance between state $\rho$ and the set of symmetric states, in terms of relative entropy. As shown in \cite{GMS09}, it turns out that the relative entropy of asymmetry is equal to the the entropy of the symmetrized version of state minus the entropy of the original state, that is 
\begin{align}\label{rel}
\Gamma(\rho)\equiv\inf_{\omega\in \text{sym}(G)}S(\rho\|\omega)= S(\rho\|\mathcal{G}(\rho))=S(\mathcal{G}(\rho))-S(\rho)\ ,
\end{align}
where  $S(\rho)=-\Tr(\rho \log\rho)$ is the von-Neuman entropy\footnote{See also \cite{vac2008} for an earlier study of  function $\Gamma(\rho)=S(\mathcal{G}(\rho))-S(\rho)$.}.  Here, the superoperator $\mathcal{G}$ is the uniform twirling  over the group $G$, which projects its input to a symmetric state, and in the case of finite groups is given by
\beq\label{twr}
\mathcal{G}(\cdot)=\frac{1}{|G|} \sum_{g\in G} \mathcal{U}_g(\cdot)\ ,
\eeq
where $|G|$ is the order of group, and we have suppressed the superscript corresponding to the label of system. For continuous groups, such as U(1) or SO(3), the  sum in Eq.(\ref{twr}) is replaced by the integral over the group with uniform (Haar) measure. Also, for the group of translations generated by a Hamiltonian $H$, or more generally  any other observable, the uniform twirling can be defined as 
\beq
\mathcal{G}(\cdot)=\lim_{T\rightarrow \infty}\frac{1}{2T}\int_{-T}^Tdt\ e^{-i H t}(\cdot) e^{i H t}=\sum_{n} \Pi_n(\cdot)\Pi_n\ ,
\eeq
where $\{\Pi_n\}_n$ are projectors to the eigen-subspaces of $H$. Therefore, in this case the uniform twirling  is equal to the \emph{dephasing} map, that is the map that dephases its input relative to the eigenbasis of $H$.\footnote{In this special case, the relative entropy of asymmetry, is called the \emph{relative entropy of coherence} by  \cite{Coh_Plenio}.} Note that the uniform twirling $\mathcal{G}$ is a special case of the \emph{resource destroying} maps recently introduced in \cite{liu2016theory}, which leave resource-free states (in this case symmetric states) unchanged and erase the resource (in this case asymmetry) in all other states.

In addition to being a measure of asymmetry, relative entropy of asymmetry has some other useful properties: (i) It is convex, i.e. $\Gamma(p\rho_1+(1-p)\rho_2)\le p\Gamma(\rho_1)+(1-p)\Gamma(\rho_2)$ for all $0\le p\le 1$. (ii) The fact that the representation of a group element $g\in G$ on the joint system $X_1$ and $X_2$  is given by $\mathcal{U}^{X_1}_g\otimes \mathcal{U}^{X_2}_g$ implies that for any group $G$ and state $\omega$ in a finite-dimensional Hilbert space, $\Gamma(\omega^{\otimes n})$ increases at most logarithmically with $n$. (iii) For a finite group $G$, $\Gamma$ is bounded by $\log|G|$. 

\section{Achievable accuracy in catalytic transformations}
In this section we find lower and upper bounds on the achievable accuracy in the catalytic transformations in Eqs.(\ref{main2},\ref{main}),  in terms of the increase in the relative entropy of asymmetry in the ideal transformation. 

\subsection{Upper bound on accuracy}\label{Sec:upper}
Using covariant channels we can only implement state transitions which do not require increase of the relative entropy of asymmetry.  Furthermore, if the asymmetry of the desired output state is much larger than the asymmetry of the input, then the transformation cannot be implemented, even approximately.  In the Supplementary Material, using  the Fannes-Audenaert inequality \cite{Audenaert:07}, we prove the following: Suppose we want to transform state $\rho^X$ of system $X$ to a state close to state $\omega^Y$ of system $Y$. Then, for any arbitrary G-covariant channel $\mathcal{E}^{X\rightarrow Y}$ from system $X$ to system $Y$, let $\epsilon=\|\mathcal{E}^{X\rightarrow Y}(\rho^X)- \omega^Y\|_1$ be the trace distance between the actual output $\mathcal{E}^{X\rightarrow Y}(\rho^X)$  and the desired output $ \omega^Y$, where   $\|\cdot\|_1$ is the sum of the singular values of the operator. Then, either $\epsilon\ge 1$, or
\beq\label{lower}
 \epsilon\times \log D_Y + 2 H(\frac{\epsilon}{2}) \ge {\Delta \Gamma}   \ ,
\eeq
where  $\Delta \Gamma=\Gamma(\omega^Y)-\Gamma(\rho^X)$ is the increase of the relative entropy of asymmetry in the desired transition, $D_Y$ is the rank of  $\mathcal{G}(\omega^Y)$, and $H(x)=-x\log x-(1-x)\log(1-x)$ is the binary entropy function. 
Using the bound $\epsilon+ 2 H(\frac{\epsilon}{2})\le 3\sqrt{\epsilon}$ for $\epsilon\in(0,1)$, this leads to the following explicit lower bound on the achievable accuracy
\beq\label{lower2}
\sqrt{\|\mathcal{E}^{X\rightarrow Y}(\rho^X)- \omega^Y \|_1 }\ge \frac{{\Delta\Gamma}}{3\log D_Y}\ .
\eeq
Eqs.(\ref{lower}) and (\ref{lower2})  can be translated to an upper bound on the fidelity of $\mathcal{E}^{X\rightarrow Y}(\rho^X)$ and $\omega^Y$, where the fidelity is defined by $F(\rho_1,\rho_2)\equiv\|\sqrt{\rho_1}\sqrt{\rho_2}\|_1$, and satisfies  $2(1-F(\rho_1,\rho_2))\le \|\rho_1-\rho_2\|_1$. Also, note that by choosing systems $X=R$ and $Y=RS$, we can apply these bounds   to the special case of catalytic transformations in the form of Eqs.(\ref{main2},\ref{main}).


\subsection{Lower bound on accuracy}\label{Sec:lower}

We saw that if a state transformation requires a large increase of the relative entropy of asymmetry, then it cannot be implemented with high fidelity using covariant channels. On the other hand, even the transitions  in which  this quantity does not increase  may also be forbidden under covariant operations. We will see  examples where  $\Gamma(\rho^X)$ is much larger than    $\Gamma(\omega^Y)$ and still the transition $\rho^X\rightarrow\omega^Y$ cannot be implemented  with covariant channels, even approximately. This should be expected because different states can break a given symmetry in different ways, i.e. they may break some subgroups but not the others. However, a single measure of asymmetry cannot see this difference (Note that this is not the only reason that a transition might be forbidden. Indeed, even for (cyclic) groups of prime order, which do not have non-trivial subgroups,  a single measure of asymmetry does not contain  enough information to determine if a state transition is possible or not.).   

However, perhaps surprisingly, it turns out that in the case of catalytic transitions described in Eq.(\ref{main}), or equivalently Eq.(\ref{main2}), just by looking to a single measure of asymmetry, namely the relative entropy of asymmetry, we can determine if  the transition can be approximately implemented with covariant operations.  
To show this, we use some recent results on approximate recoverability \cite{fawzi2015quantum} to find a lower bound on the achievable fidelity of implementing the catalytic  transformations in  Eq.(\ref{main2}) and Eq.(\ref{main}). See also Sec.\ref{Sec:revers} on reversibility, for a more general proof of this lower bound. 

Consider the tripartite state 
\begin{align}
\Sigma^{CRS}&=\frac{1}{|G|}\sum_{g\in G} |g\rangle\langle g|^C\otimes  \mathcal{U}^{R}_g(\tau^{R}) \otimes \mathcal{U}^{S}_g(\sigma^{S})\ ,
\end{align}
where $C$, or the $\emph{classical}$ background reference frame, is a system with the Hilbert space spanned by the orthogonal states $\{|g\rangle:g\in G\}$. In the case of continuous groups such as U(1) or SO(3) we replace the sum with an integral over group with uniform (Haar) measure. 
Let $\Sigma^{CR}=\Tr_S(\Sigma^{CRS})$  and $\Sigma^{RS}=\Tr_C(\Sigma^{CRS})$  be  the reduced state of the joint systems $CR$ and $RS$, respectively. Then, it can be easily shown that the quantum mutual information between system $R$ (or $S$) and $C$ is equal to $\Gamma(\tau^R)$ (or $\Gamma(\sigma^S)$). 
Furthermore, the conditional mutual information  
\beq
I(S:C|R)_\Sigma=S(\Sigma^{CR})+S(\Sigma^{RS})-S(\Sigma^{CRS})-S(\Sigma^{R})\ ,
\eeq
turns out to be equal to
\beq\label{eq1}
I(C:S|R)_\Sigma=\Gamma(\tau^R\otimes \sigma^S)-\Gamma(\tau^R)\equiv \Delta\Gamma , 
\eeq
that is the difference between the relative entropy of asymmetry for states $\tau^R\otimes \sigma^S$ and $\tau^R$.  

According to a recent result of Fawzi and Renner \cite{fawzi2015quantum} (See also \cite{wilde2015recoverability, berta2015renyi, sutter2016strengthened}) if the conditional mutual information $I(C:S|R)_\Sigma$ is small then state $\Sigma^{CRS}$ can  be approximately reconstructed from the reduced state $\Sigma^{CR}$ in the following sense:  there exists a quantum channel  $\mathcal{R}^{R\rightarrow RS}$ which maps system $R$ to systems $RS$ such that  
 \beq\label{FR}
\text{F}(\mathcal{R}^{R\rightarrow RS}(\Sigma^{CR}), \Sigma^{CRS}) \ge 2^{ -\frac{1}{2}I(C:S|R)_\Sigma }\ ,
\eeq
where we have suppressed the identity super-operator which acts on system $C$. 
The recovery map which satisfies Eq.(\ref{FR}) can be chosen to be  a rotated version of the Petz recovery map \cite{fawzi2015quantum, wilde2015recoverability, berta2015renyi}.

Then, the reconstructed state has the following form 
\beq
\mathcal{R}^{R\rightarrow RS}(\Sigma^{CR})=\frac{1}{|G|}\sum_{g\in G} |g\rangle\langle g|^C\otimes  \mathcal{R}^{R\rightarrow RS}\circ\mathcal{U}^{R}_g(\tau^{R}) \ .
\eeq
The fidelity of this state with state $\Sigma^{CRS}$ is given by
\bes
\begin{align*}
&\text{F}(\mathcal{R}^{R\rightarrow RS}(\Sigma^{CR}), \Sigma^{CRS})\\ &=\frac{1}{|G|}\sum_{g\in G} F(\mathcal{R}^{R\rightarrow RS}\circ\mathcal{U}^{R}_g(\tau^{R}) ,  \mathcal{U}^{RS}_g(\tau^{R}\otimes \sigma^{S}))\\ &=\frac{1}{|G|}\sum_{g\in G} F(\mathcal{U}^{RS}_{g^{-1}}\circ \mathcal{R}^{R\rightarrow RS}\circ\mathcal{U}^{R}_g(\tau^{R}), \tau^{R}\otimes\sigma^{S})\ ,
\end{align*}
\ees
where we have used the notation $\mathcal{U}^{RS}_{g}=\mathcal{U}^{R}_{g}\otimes\mathcal{U}^{S}_{g}$. 
Here the first equality follows from the orthogonality of states $\{|g\rangle:g\in G\}$, and the second equality follows from the invariance of fidelity under unitary transformations.  
Next, using the  concavity of fidelity \cite{nielsen2000quantum}, we find that the average fidelity in the last line  is less than or equal to the fidelity of the averaged states. That is $\text{F}(\mathcal{R}^{R\rightarrow RS}(\Sigma^{CR}), \Sigma^{CRS})$ is less than or equal to 
\bes
\begin{align}
&F(\frac{1}{|G|}\sum_{g\in G} \mathcal{U}^{RS}_{g^{-1}}\circ \mathcal{R}^{R\rightarrow RS}\circ\mathcal{U}^{R}_g(\tau^{R}), \tau^{R}\otimes\sigma^{S})\\  &= F(\mathcal{R}_\text{sym}^{R\rightarrow RS}(\tau^{R}), \tau^{R}\otimes\sigma^{S})\ ,
\end{align}
\ees
where  $\mathcal{R}_\text{sym}^{R\rightarrow RS}\equiv \frac{1}{|G|}\sum_{g\in G} \mathcal{U}^{RS}_{g^{-1}}\circ \mathcal{R}^{R\rightarrow RS}\circ\mathcal{U}^{R}_g$
is the symmetrized version of $\mathcal{R}_\text{sym}^{R\rightarrow RS}$, which satisfies  the covariance condition in Eq.(\ref{cov}). Using this together with Eq.(\ref{eq1}) and Eq.(\ref{FR}) we conclude that 
\begin{theorem}\label{Thm}
There exists a covariant channel $\mathcal{E}^{R\rightarrow RS}$, i.e. a channel satisfying Eq.(\ref{cov}), which transforms $\tau^{R}$ to a state whose fidelity with the desired state  $\tau^{R}\otimes\sigma^{S}$ is lower bounded by $2^{-\frac{\Delta \Gamma}{2}}$, that is
\beq\label{abs}
 F(\mathcal{E}^{R\rightarrow RS}(\tau^{R}), \tau^{R}\otimes\sigma^{S})\ge 2^{-\frac{\Delta \Gamma}{2}  } \ ,
\eeq
where $\Delta \Gamma=\Gamma(\tau^R\otimes \sigma^S)-\Gamma(\tau^R)$ is the increase in the relative entropy of asymmetry in the ideal process $\tau^{R}\rightarrow\tau^{R}\otimes\sigma^{S}$. 
\end{theorem}
Indeed, as we explain in the Supplementary Material (See Sec. \ref{sec:reversibility}), it follows from the results of \cite{wilde2015recoverability, berta2015renyi} that the covariant map $\mathcal{E}^{R\rightarrow RS}$  can be chosen to be the Petz Recovery map, up to some covariant unitary rotations in the input and output, i.e.
\beq
\mathcal{E}^{R\rightarrow RS}=\mathcal{V}_{-t}^{RS}\circ\mathcal{R}^{R\rightarrow RS}_P\circ \mathcal{V}_t^{R}\ ,
\eeq
where $\mathcal{R}^{R\rightarrow RS}_P$ is the Petz recovery map
\begin{align}\label{Petz}
&\mathcal{R}^{R\rightarrow RS}_P(\omega^R)=\nonumber\\
&\sqrt{\mathcal{G}(\tau^R\otimes\sigma^S)}(\frac{1}{\sqrt{\mathcal{G}(\tau^R)}}\omega^R\frac{1}{\sqrt{\mathcal{G}(\tau^R)}} \otimes I^S)   \sqrt{\mathcal{G}(\tau^R\otimes\sigma^S)}\ .
\end{align}
Here $I^S$ is the identity operator on system $S$, and the inverse of $\sqrt{\mathcal{G}(\tau^R)}$ is defined only on the support of this operator.  Furthermore, $\mathcal{V}_t^{R}$ and $\mathcal{V}_{-t}^{RS}$ are  unitary covariant  channels acting on the input and output systems $R$ and $RS$ respectively, defined by
\bes
\begin{align}
\mathcal{V}_t^{R}(\omega^R)&\equiv[\mathcal{G}(\tau^R)]^{it}\ \omega^R\  [\mathcal{G}(\tau^R)]^{-it} ,  \\
\mathcal{V}_{-t}^{RS}(\omega^{RS})&\equiv[\mathcal{G}(\tau^{RS})]^{-it}\ \omega^{RS}\ [\mathcal{G}(\tau^{RS})]^{it}\ 
\end{align}
\ees
 for some unknown $t\in\mathbb{R}$. Note that for any  $t\in\mathbb{R}$ the support of state $\mathcal{V}_t^{R}(\tau^R)$ is contained in the support of $\sqrt{\mathcal{G}(\tau^R)}$, and therefore the action of $\mathcal{R}^{R\rightarrow RS}_P$ is well-defined on this state.

 This theorem together with upper bound in Eq.(\ref{lower}) imply that the quantity $\Delta \Gamma$ mainly determines weather the catalytic transition in Eq.(\ref{main}) can be implemented approximately or not.  Note that because of the monotonicity of fidelity under partial trace, after applying the map $\mathcal{E}^{R\rightarrow RS}$ the fidelity of the reduced state of system $S$ with its desired state $\sigma^{S}$ is larger than or equal to both sides of Eq.(\ref{abs}).

In the following we consider some examples of applications of this result.

\section{Applications}\label{Sec:app}

\subsection{State preparation using quantum clocks}

As a simple model for a clock we consider a harmonic oscillator with Hamiltonian $H^R=\omega (1/2+N)$, where $N=\sum_{n=0}^\infty n |n\rangle\langle n|$ is the number operator, and $\omega$ is the frequency of oscillation. 
Assume the system $S$ is another harmonic oscillator with the same frequency, which is initially in a time-independent state,   such as the thermal state. Suppose we want to prepare the system  $S$ in the time-dependent state $|+^S\rangle=(|0\rangle+|1\rangle)/\sqrt{2}$ by coupling it to the clock $R$ (Note that $|0\rangle$ and $|1\rangle$ are, respectively, the ground and the first excited states of the harmonic oscillator.). 

We consider two different initial states of clock, i.e. 
\beq
|\psi^R_{k}\rangle=\frac{1}{\sqrt{T}}\sum_{n=0}^{T-1} |k n\rangle , \ \ \ \  k=1,2\ .
\eeq
Note that $\Gamma(|\psi^R_{1,2}\rangle)=\log T$ and $\Gamma(|+^S\rangle)=1$. Therefore, for large enough $T$,   $\Gamma(|\psi^R_{1,2}\rangle)$ could be  arbitrary larger than $\Gamma(|+^S\rangle)$. Despite this, it turns out that  using state $|\psi^R_2\rangle$ and operations which are covariant under time translations we can never prepare $S$ in a state close to $|+^S\rangle$, that is a state whose fidelity with the desired state $|+^S\rangle$ is larger than the fidelity of the time-independent state $|0^S\rangle$ with this state.\footnote{ This can be seen, e.g., using the fact that the input $|\psi^R_2\rangle$ is invariant under the time translation $e^{-i H^R \pi/\omega}$, and therefore for any time-invariant channel the output should also be invariant under this time translation. But under the time translation $e^{i H^S \pi/\omega}$ the state $|+^S\rangle$ is mapped to an orthogonal state. It follows that regardless of how large is $T$, with a clock in state $|\psi^R_2\rangle$ we can never prepare state $|+^S\rangle$, or something close to it.}  On the other hand, having the clock $R$ in  state $|\psi^R_{1}\rangle$ we can prepare state $|\psi^R_{1}\rangle|+^S\rangle$ with fidelity $1-2/T$ \cite{marvian2008building}. This can also be shown  using theorem \ref{Thm}:  first, note that 
\beq
\Gamma(|\psi^R_{1}\rangle|+^S\rangle)=\log T+\frac{1}{T}\ ,\ \ \   \Gamma(|\psi^R_{2}\rangle|+^S\rangle)=\log T+1 \nonumber .
\eeq
Therefore, while  in the transition $|\psi^R_{2}\rangle\rightarrow |\psi^R_{2}\rangle|+^S\rangle$ the increase in $\Gamma$ remains independent of $T$, that is $\Delta \Gamma=1$, in the transition   $|\psi^R_{1}\rangle\rightarrow |\psi^R_{1}\rangle|+^S\rangle$ the increase in $\Gamma$ is $1/T$, and hence vanishes for large $T$. Using theorem \ref{Thm} we conclude that there exists a quantum channel which is invariant under time translations and  implements the latter transition  with fidelity larger than or equal to $2^{-1/(2T)}$, which   for $T\gg1$ is $\approx 1- 1/(2 T \log e)$.

\subsection{Petz recovery map as the optimal cloner}
As the next example, we consider the application of our result in the case of cloners. Consider a cloner that receives $n$ copies of an unknown state $\psi\in \mathbb{C}^d$, and generates an output in state close to $n+k$ copies of this state. As we saw before, the optimal cloner can be chosen to satisfy the covariance condition in Eq.(\ref{cloner}). Therefore, in this case the relevant symmetry is $G=\text{U}(d)$, the group of unitaries acting on $\mathbb{C}^d$. Then, using the fact that for any integer $r$, $U^{\otimes r}$ acts irreducibly on the symmetric subspace of  $(\mathbb{C}^d)^{\otimes r}$, we find that $\mathcal{G}(\psi^{\otimes r})= \Pi_+^{(r)}/d_+(r)$, where $\Pi_+^{(r)}$ is the projector to the symmetric subspace, whose dimension is $d_+(r)={{r+d-1}\choose{d-1}}$. It follows that in the desired transformation $\psi^{\otimes n}\rightarrow \psi^{\otimes (n+k)}$, the increase in the relative entropy of asymmetry is $\Delta \Gamma=\log \frac{d_+(n+k)}{d_+(n)}$. Therefore, by virtue of theorem \ref{Thm}, we find that this state transition can be implemented   using a covariant channel, with fidelity  larger than or equal to $\sqrt{\frac{d_+(n)}{d_+(n+k)}}$. This is exactly the optimal fidelity of cloning, as  shown by Werner \cite{werner1998optimal}.  

Furthermore, because $\mathcal{G}(\psi^{\otimes n})$ and $\mathcal{G}(\psi^{\otimes (n+k)})$ are both completely mixed states on the symmetric subspaces of the input and output Hilbert spaces, it follows that for any input state  $\psi^{\otimes n}$, the unitary channels  $\mathcal{V}_t^{R}$ and $\mathcal{V}_{-t}^{RS}$ act trivially,  and therefore  the covariant channel that achieves this fidelity is the Petz recovery map in Eq.(\ref{Petz}) itself, i.e.  
\begin{align*}
\mathcal{E}^{n\rightarrow n+k}(\psi^{\otimes n})&=\frac{d_+(n)}{d_+(n+k)} \Pi_+^{(n+k)}[|\psi\rangle\langle\psi|^{\otimes n} \otimes I^{\otimes k}]\Pi_+^{(n+k)} .
\end{align*}
Interestingly, this is  exactly the optimal cloning map which was originally found by Werner \cite{werner1998optimal}, and maximizes the fidelity of cloning.


\section{Reversibility} \label{Sec:revers}
The problem of catalytic transitions under covariant transformations is indeed closely related to a more general problem, which is of independent interest: How well we can reverse the effect of a covariant channel on a given state using only covariant channels? Roughly speaking, one expects that because asymmetry is the resource that cannot be generated by covariant channels,  if under the effect of the first covariant channel the amount of asymmetry of state does not decrease considerably, then the process should be approximately reversible. As we show in the Supplementary Material, using the results of \cite{fawzi2015quantum, wilde2015recoverability, berta2015renyi},  this intuition is indeed correct: 
\begin{theorem}\label{Thm2}
Suppose under the action of a covariant channel $\mathcal{F}^{X\rightarrow Y}$ the relative entropy of asymmetry of a given state $\rho^X$ drops by  $\Delta \Gamma=\Gamma(\rho^X)-\Gamma(\mathcal{F}^{X\rightarrow Y}(\rho^X))$. Then there exists a covariant channel $\mathcal{R}^{Y\rightarrow X}$, such that $F(\mathcal{R}^{Y\rightarrow X}\circ\mathcal{F}^{X\rightarrow Y}(\rho^X), \rho^X)\ge 2^{-\Delta \Gamma/2}$. 
\end{theorem}
Note that choosing $X=RS$, $Y=R$, and the map $\mathcal{F}^{X\rightarrow Y}$  to be the partial trace over system $S$, we can obtain theorem \ref{Thm} as a special case. In the Supplementary Material we also present similar results in the context of the resource theory of (speakable) coherence, in terms of dephasing-covariant channels \cite{marvian2016quantify, chitambar2016incoherent}. See also \cite{wehner2015work} for similar results on reversibility in the context of the resource  theory of thermodynamics.

\section{Conclusion}

One of the advantages of using the relatively abstract language of symmetry to understand physical phenomena, is that it clarifies similarities in seemingly different problems in different contexts. Hence,  intuitions from one physical phenomenon can be applied to another problem in another context.  The resource theory of asymmetry, whose central question is what \emph{can} and what \emph{cannot} be done with symmetric time evolutions, inherits this advantage.  Formalizing a result in the language of this resource theory, can help to clarify its applications in a wide range of problems, from clocks to cloners.

\bibliography{Ref_v1}

\onecolumngrid

\newpage

\maketitle
\vspace{-5in}
\begin{center}

\Large{\textbf{Supplementary Material}}
\end{center}
\appendix

\section{Reversibility under covariant operations}\label{sec:reversibility}
In this section we prove theorem \ref{Thm2},  which indeed can be thought as a generalization of theorem \ref{Thm}. 
Theorem  \ref{Thm2} is basically a corollary of the following result from \cite{wilde2015recoverability, berta2015renyi}, which is a  generalization of Fawzi and Renner original result:


\begin{theorem}\label{FRW}(From \cite{wilde2015recoverability, berta2015renyi})
Consider a channel $\mathcal{N}$ and a pair of state $\rho$ and $\kappa$ where the support of $\rho$ is contained in the support of $\kappa$. Then, there exists a recovery channel  $\mathcal{R}_\kappa$ which maps $\mathcal{N}(\kappa)$ to $\kappa$ and 
\beq
S(\rho\|\kappa)- S(\mathcal{N}(\rho)\| \mathcal{N}(\kappa)\ge -2\log F(\mathcal{R}_\kappa\circ\mathcal{N}(\rho), \rho)\ ,
\eeq
This cannel can be chosen to be
\beq\label{Rec_gen}
\mathcal{R}_{\kappa}= \mathcal{V}_{\kappa,t} \circ\mathcal{R}_{\kappa,P}\circ \mathcal{V}_{\mathcal{N}(\kappa), -t}
\eeq
for some $t\in\mathbb{R}$, where $\mathcal{R}_{\kappa,P}$ is the Petz recovery map
\beq\label{Petz_app}
\mathcal{R}_{\kappa,P}(\rho)=\sqrt{\kappa}\ \mathcal{N}^\dag\big(\frac{1}{\sqrt{\mathcal{N}(\kappa)}}\ \rho\ \frac{1}{\sqrt{\mathcal{N}(\kappa)}}\big)\sqrt{\kappa}\ .
\eeq
and 
\beq
\mathcal{V}_{\omega,t}(\cdot)=\omega^{i t}(\cdot) \omega^{-i t}\ . 
\eeq
\end{theorem}
Note that $\mathcal{N}^\dag$ is the adjoint of $\mathcal{N}$ defined via $\Tr(\mathcal{N}^\dag(X) Y)=\Tr(X \mathcal{N}(Y))$.

\begin{proof} (Theorem \ref{Thm2})

Suppose  in theorem \ref{FRW} we choose  channel $\mathcal{N}$ to be $\mathcal{F}^{X\rightarrow Y}$. In the following we suppress the superscripts $X$ and $Y$. By assumption this channel is G-covariant, i.e. satisfies  
\beq\label{cov-app}
\mathcal{F}\circ \mathcal{U}_g=\mathcal{U}_g\circ \mathcal{F}\ \ : \ \forall g\in G\ .
\eeq
Averaging\footnote{In the case of the group of translations $\{e^{-i H t}:t\in \mathbb{R}\}$ we can average over a finite interval $[-T,T]$, and then look at the limit of $T\rightarrow \infty$. } over group $G$ we find
\beq\label{qwer}
\mathcal{F}\circ \mathcal{G}= \mathcal{G}\circ \mathcal{F}\ .
\eeq
This implies
\begin{align}
S(\rho\|\mathcal{G}(\rho))- S(\mathcal{F}(\rho)\| \mathcal{F}\circ\mathcal{G}(\rho) )&=S(\rho\|\mathcal{G}(\rho))- S(\mathcal{F}(\rho)\| \mathcal{G}\circ\mathcal{F}(\rho) )\\ &= \Gamma(\rho)-\Gamma(\mathcal{F}(\rho))\equiv \Delta \Gamma\ ,
\end{align}
where the first equality follows from Eq.(\ref{qwer}). Therefore, from theorem \ref{FRW} we know  the  recovery channel $\mathcal{R}_\kappa$ defined by Eq.(\ref{Rec_gen}) satisfies 
\beq
\Delta\Gamma= S(\rho\|\mathcal{G}(\rho))- S(\mathcal{F}(\rho)\| \mathcal{F}\circ \mathcal{G}(\rho)\ge -2\log F(\mathcal{R}_\kappa\circ\mathcal{F}(\rho), \rho)\ ,
\eeq
where $\kappa=\mathcal{G}(\rho)$.  Next, we show that if $\mathcal{F}$ is G-Covariant then the recovery channel 
\beq
\mathcal{R}_{\kappa}= \mathcal{V}_{\kappa,t} \circ\mathcal{R}_{\kappa,P} \circ \mathcal{V}_{\mathcal{F}(\kappa), -t}\ ,
\eeq
for $\kappa=\mathcal{G}(\rho) $ will also be G-Covariant. First, note that $G$-covariance of $\mathcal{F}$ implies  $G$-covariance of $\mathcal{F}^\dag$. This can be seen by considering the adjoint  of both sides of Eq.(\ref{cov-app}), which implies
\beq
\mathcal{U}_{g^{-1}}\circ \mathcal{F}^\dag=\mathcal{F}^\dag\circ \mathcal{U}_{g^{-1}} \ : \ \forall g\in G\ ,
\eeq
and hence $\mathcal{U}_{g}\circ \mathcal{F}^\dag=\mathcal{F}^\dag\circ \mathcal{U}_{g}$ for all $g\in G$. Then, we note that both states $\kappa=\mathcal{G}(\rho)$ and  $\mathcal{F}(\kappa)=\mathcal{F}\circ\mathcal{G}(\rho)= \mathcal{G}\circ \mathcal{F}(\rho)$ are G-invariant. It follows that all the channels 
\beq\label{ew1}
\mathcal{V}_{\kappa,t}(\omega)=[\mathcal{G}(\rho)]^{it} (\omega) [\mathcal{G}(\rho)]^{-it} 
\eeq
and
\beq\label{ew2}
\mathcal{V}_{\mathcal{F}(\kappa),-t}(\omega)= [\mathcal{F}\circ\mathcal{G}(\rho)]^{-it} (\omega) [\mathcal{F}\circ\mathcal{G}(\rho)]^{it} 
\eeq
and
\beq
\mathcal{R}_{\kappa,P}(\omega)=\sqrt{ \mathcal{G}(\rho)}\ \mathcal{F}^\dag\big(\frac{1}{\sqrt{ \mathcal{G}\circ\mathcal{F}(\rho)}}\ \omega\ \frac{1}{\sqrt{ \mathcal{G}\circ\mathcal{F}(\rho)}}\big)\sqrt{ \mathcal{G}(\rho)}
\eeq
 are G-covariant, and so is their combination 
 \beq\label{explicit}
 \mathcal{R}_{\kappa}= \mathcal{V}_{\kappa,t} \circ\mathcal{R}_{\kappa,P} \circ \mathcal{V}_{\mathcal{F}(\kappa), -t}\ .
 \eeq
  This completes the proof.
 
\end{proof}


\section{General upper bound on accuracy (proof of Eq.(\ref{lower}))}

Recall the Fannes-Audenaert  \cite{Audenaert:07} inequality  
\beq\label{Fan}
|S(\rho_1)-S(\rho_2)|\le \log D \times \frac{ \|\rho_1-\rho_2\|_1}{2}+H(\frac{ \|\rho_1-\rho_2\|_1}{2})\ ,
\eeq
where  $\rho_1$ and $\rho_2$ are arbitrary pair of density operators in a $D$-dimensional Hilbert space, and $H(x)=-x\log x-(1-x)\log(1-x)$.  Using the fact the uniform twirling $\mathcal{G}$ is a quantum channel, and the trace norm is non-increasing under quantum channels, we have 
\beq
\|\mathcal{G}(\rho_1)-\mathcal{G}(\rho_2)\|_1\le \|\rho_1-\rho_2\|_1\ .
\eeq 
Then, assuming $\|\rho_1-\rho_2\|_1\le 1$, and using the fact that function $H$ is monotonically increasing in the interval $(0,1/2)$ we find 
\bes
\begin{align}
|S(\mathcal{G}(\rho_1))-S(\mathcal{G}(\rho_2))|&\le \log D \times \frac{ \|\mathcal{G}(\rho_1)-\mathcal{G}(\rho_2)\|_1}{2}+H(\frac{ \|\mathcal{G}(\rho_1)-\mathcal{G}(\rho_2)\|_1}{2})\ \\ &\le \log D \times  \frac{\|\rho_1-\rho_2\|_1}{2}+H(\frac{ \|\rho_1-\rho_2\|_1}{2})\ .
\end{align}
\ees
This bound together  with bound (\ref{Fan}) and the triangle inequality imply
\bes\label{gen}
\begin{align}
\Gamma(\rho_1)-\Gamma(\rho_2)&=[S(\mathcal{G}(\rho_1))-S(\rho_1)]-[S(\mathcal{G}(\rho_2))-S(\rho_2)]\\ &\le |S(\rho_2)-S(\rho_1)|+|S(\mathcal{G}(\rho_1))-S(\mathcal{G}(\rho_2))|  \\ &\le \log D \times  \|\rho_1-\rho_2\|_1+2 H(\frac{ \|\rho_1-\rho_2\|_1}{2})
\end{align}
\ees
Applying this bound for $\rho_1=\omega^Y$ and $\rho_2=\mathcal{E}^{X\rightarrow Y}(\rho^X)$, and using the fact that since  $\mathcal{E}^{X\rightarrow Y}$ is G-Covariant then  $\Gamma(\mathcal{E}^{X\rightarrow Y}(\rho^X))\le \Gamma(\rho^X)$,   we find 
\beq
\Gamma(\omega^Y)-\Gamma(\rho^X)  \le \Gamma(\omega^Y)-\Gamma(\mathcal{E}^{X\rightarrow Y}(\rho^X)\le \log D^\ast_Y \times  \|\omega^Y-\mathcal{E}^{X\rightarrow Y}(\rho^X) \|_1+2 H(\frac{ \|\omega^Y-\mathcal{E}^{X\rightarrow Y}(\rho^X) \|_1}{2})\ ,
\eeq
where $D^\ast_Y$ is the dimension of system $Y$. This is a weaker version of Eq.(\ref{lower}), because in general, $D_Y\le D_Y^\ast$, that is the rank of $\mathcal{G}(\omega^Y)$ is less than or equal to the dimension of system $Y$. Next, we prove bound (\ref{lower}).

Let $\Pi$ be the projector to the support of $\mathcal{G}(\omega^Y)$. Using the fact that  $\mathcal{G}(\tau^Y)$ commutes with $U^Y(g)$, for all $g\in G$, we find that $\Pi$ also commutes with  $U^Y(g)$, for all $g\in G$. Define the channel 
\beq
\mathcal{L}^{Y\rightarrow Y}(\omega^Y)=\Pi\ \omega^Y\ \Pi+\Tr(\Pi\omega^Y) \frac{\Pi}{\Tr(\Pi)}\ . 
\eeq
This is the quantum channel which projects any input state to a state restricted to the support of $\mathcal{G}(\omega^Y)$, and if the input state is found to be outside this subspace, then it prepares the totally mixed state in the support of  $\mathcal{G}(\omega^Y)$. Using the fact that $\Pi$  commutes with  $U^Y(g)$ for all $g\in G$ we can easily see that $\mathcal{L}^{Y\rightarrow Y}$ is a covariant channel. This implies 
 \begin{align*}
\Gamma(\omega^Y)-\Gamma(\rho^X)&\le \Gamma(\omega^Y)-\Gamma(\mathcal{E}^{X\rightarrow Y}(\rho^X))\\ &\le \Gamma(\omega^Y)-\Gamma(\mathcal{L}^{Y\rightarrow Y}\circ\mathcal{E}^{X\rightarrow Y}(\rho^X))\ ,
\end{align*}
where both bounds follow  from the monotonicity of $\Gamma$ under covariant channels.  Next, note that the support of both density operators $\omega^Y$ and $\mathcal{L}^{Y\rightarrow Y}\circ\mathcal{E}^{X\rightarrow Y}(\rho^X)$ are contained in the support of $\mathcal{E}^{X\rightarrow Y}(\rho^X)$, whose dimension is denoted by $D_Y$. Therefore, we can use Eq.(\ref{gen}) with $D=D_Y$. Then, we find
\bes
\begin{align}
\Delta\Gamma&\equiv \Gamma(\omega^Y)-\Gamma(\rho^X)\\ &\le \Gamma(\omega^Y)-\Gamma(\mathcal{E}^{X\rightarrow Y}(\rho^X))\\ &\le \Gamma(\omega^Y)-\Gamma(\mathcal{L}^{Y\rightarrow Y}\circ\mathcal{E}^{X\rightarrow Y}(\rho^X) )\\ &\le  \log D_Y \times  \|\omega^Y-\mathcal{L}^{Y\rightarrow Y}\circ\mathcal{E}^{X\rightarrow Y}(\rho^X) \|_1+2 H(1/2\times\|\omega^Y-\mathcal{L}^{Y\rightarrow Y}\circ\mathcal{E}^{X\rightarrow Y}(\rho^X))
\\ &=  \log D_Y \times  \|\mathcal{L}^{Y\rightarrow Y}(\omega^Y)-\mathcal{L}^{Y\rightarrow Y}\circ\mathcal{E}^{X\rightarrow Y}(\rho^X) \|_1+2 H(1/2\times\|\mathcal{L}^{Y\rightarrow Y}\circ(\omega^Y)-\mathcal{L}^{Y\rightarrow Y}\circ\mathcal{E}^{X\rightarrow Y}(\rho^X)\|_1)
\\ &\le  \log D_Y \times  \|\omega^Y-\mathcal{E}^{X\rightarrow Y}(\rho^X)\|_1+2 H(1/2\times \|\omega^Y-\mathcal{E}^{X\rightarrow Y}(\rho^X)\|_1)\ ,
\end{align}
\ees
where to get the fourth line we have used Eq.(\ref{gen}), and to get the fifth line we have used  the fact that the support of $\omega^Y$ is contained in the support of $\mathcal{G}(\omega^Y)$, and therefore the channel $\mathcal{L}^{Y\rightarrow Y}$ leaves $\omega^Y$ invariant, and to get the last line we have used monotonicity of the trace distance under quantum channels. This completes the proof of Eq.(\ref{lower}).

\section{Recoverability for  dephasing-covariant channels}
In section \ref{sec:reversibility}, we used theorem  \ref{FRW} to prove that if under a covariant channel the relative entropy of asymmetry does not drop considerably, then the process can be approximately reversed using a covariant channel. In this section we use theorem  \ref{FRW} again to show that a similar theorem holds in the case of \emph{dephasing-covariant} channels  \cite{marvian2016quantify, chitambar2016incoherent}.

Consider a complete set of orthogonal projectors $\{P_j\}_j$. Then, the \emph{dephasing map} $\mathcal{D}$ is defined by
\beq
\mathcal{D}(\cdot)=\sum_j P_j(\cdot)P_j\ .
\eeq
A quantum channel $\mathcal{F}$ is called \emph{dephasing-covariant} \cite{marvian2016quantify, chitambar2016incoherent}, if it satisfies
\beq\label{def_deph}
\mathcal{F}\circ \mathcal{D}=\mathcal{D}\circ \mathcal{F}\ .
\eeq
Note that here we are only considering quantum channels whose input and output spaces are the same.

The \emph{relative entropy of coherence} \cite{Coh_Plenio} is defined by  
\beq
\Lambda(\rho)=\inf_{\omega\in \mathcal{I}}S(\rho\|\omega)=S(\rho\|\mathcal{D}(\rho))= S(\mathcal{D}(\rho))-S(\rho)\ ,
\eeq  
where $\mathcal{I}$  is the set of \emph{incoherent} states relative this basis, i.e. states that can be written as $\sum_j p_j P_j$ for a probability distribution $p_j$. It can be easily shown that the relative entropy of coherence is non-increasing under dephasing-covariant channels \cite{marvian2016quantify, chitambar2016incoherent}. 

\begin{theorem}\label{Thm:speak}
Suppose, under the action of dephasing-covariant channel $\mathcal{F}$ the relative entropy of coherence of state $\rho$ drops by $\Delta \Lambda$, i.e.
\beq
\Delta \Lambda=\Lambda(\rho)-\Lambda(\mathcal{F}(\rho))\ .
\eeq 
Then, there exists a dephasing-covariant channel $\mathcal{R}$ such that
\beq
F(\mathcal{R}\circ\mathcal{F}(\rho),\rho)\ge 2^{-\Delta\Lambda/2}\ .
\eeq
\end{theorem}
\begin{proof}
The result follows from theorem \ref{FRW} for the special case where $\mathcal{N}=\mathcal{F}$ is a dephasing-covaraint channel, and $\kappa=\mathcal{D}(\rho)$. 

First, note that
\beq
S(\rho\|\mathcal{D}(\rho))-S(\mathcal{F}(\rho)\|\mathcal{F}\circ\mathcal{D}(\rho))=S(\rho\|\mathcal{D}(\rho))-S(\mathcal{F}(\rho)\|\mathcal{D}\circ\mathcal{F}(\rho))=\Delta\Lambda\ ,
\eeq
where the first equality follows from the fact that channel $\mathcal{F}$ is Dephasing-covariant.

Therefore, from theorem \ref{FRW} we know  the  recovery channel $\mathcal{R}_\kappa$ defined by Eq.(\ref{Rec_gen}) satisfies 
\beq\label{qwed}
\Delta\Lambda= S(\rho\|\mathcal{D}(\rho))- S(\mathcal{F}(\rho)\| \mathcal{F}\circ \mathcal{D}(\rho)\ge -2\log F(\mathcal{R}_\kappa\circ\mathcal{F}(\rho), \rho)\ ,
\eeq
where $\kappa=\mathcal{D}(\rho)$.  Next, we show that if we choose $\kappa=\mathcal{D}(\rho)$ and $\mathcal{N}=\mathcal{F}$, then the fact that $\mathcal{F}$ is dephasing-covariant implies that  the recovery channel $\mathcal{R}_\kappa$ is also  Dephasing-covariant.

First, by looking to the adjoint of Eq.(\ref{def_deph}) we find that if  $\mathcal{F}$ is dephasing-covaraint, then $\mathcal{F}^\dag$ is also dephasing-covariant (Note that $\mathcal{D}^\dag=\mathcal{D}$).  Combining this with the fact that states $\mathcal{D}(\rho)$, and $\mathcal{F}\circ\mathcal{D}(\rho)=\mathcal{D}\circ\mathcal{F}(\rho)$ are both incoherent we find that the Petz recovery channel, defined by Eq.(\ref{Petz_app}) is also dephasing-covariant.

Finally, using the fact that both states $\mathcal{D}(\rho)$, and $\mathcal{F}\circ\mathcal{D}(\rho)=\mathcal{D}\circ\mathcal{F}(\rho)$ are  incoherent, we find that  for all $t\in \mathbb{R}$, the operations   $\mathcal{V}_{\mathcal{D}(\rho),t}$ and  $\mathcal{V}_{\mathcal{F}\circ\mathcal{D}(\rho),t}$ defined by Eq.(\ref{ew1}) and Eq.(\ref{ew2}) are also dephasing-covariant unitary channels. Because the combination of dephasing-covaraint channels is a dephasing-covaraint channel,  we conclude that the recovery channel $ \mathcal{V}_{\mathcal{D}(\rho),t} \circ\mathcal{R}_{\mathcal{D}(\rho),P}\circ \mathcal{V}_{\mathcal{F}\circ\mathcal{D}(\rho), -t}$ which satisfies bound \ref{qwed} is also dephasing-covariant. This completes the proof.
 \end{proof}

\end{document}